\documentclass[a4paper]{article}
\usepackage{braket}
\usepackage[basic]{complexity}
\usepackage[numbers,sort&compress]{natbib}
\usepackage{tabularx}
\usepackage{graphicx,color}

\usepackage{amsmath,amssymb,amsthm}
\usepackage{subcaption}
\usepackage[ruled,vlined]{algorithm2e}
\usepackage{authblk}

\SetKwInput{KwInput}{Input}
\SetKwInput{KwOutput}{Output}

\newtheorem{theorem}{Theorem}
\newtheorem{corollary}{Corollary}
\def\R{{\mathbb{R}}}
\def\fta{\fontsize{12pt}{0}}

\begin{document}

\title{\large\bf QDNN: DNN with Quantum Neural Network Layers}
\author[1,2]{\fta Chen Zhao}
\author[1,2]{\fta Xiao-Shan Gao}
\affil[1]{\small \em Academy of Mathematics and Systems Science, Chinese Academy of Sciences, Beijing 100190, China}
\affil[2]{\small \em University of  Chinese Academy of Sciences, Beijing 100049, China}
\date{}
\maketitle

\begin{abstract}
\noindent
In this paper, we introduce a quantum extension of classical DNN, QDNN. The QDNN consisting of quantum structured layers can uniformly approximate any continuous function 
and has more representation power than the classical DNN.
It still keeps the advantages of the classical DNN such as the non-linear activation, the multi-layer structure, and the efficient backpropagation training algorithm. Moreover, the QDNN can be used on near-term noisy intermediate-scale quantum processors. A numerical experiment for image classification based on quantum DNN is given, where a high accuracy rate is achieved.
\end{abstract}

\vspace{2pc}
\noindent{\it Keywords}: Deep neural networks, quantum machine learning, hybrid quantum-classical algorithm, NISQ

\section{Introduction} 
\label{sec:introduction}
Quantum computers use the principles of quantum mechanics for computing, which are more powerful than classical computers in many computing problems \cite{shor1994algorithms, grover1996fast}.
Noisy intermediate-scale quantum (NISQ) \cite{preskill2018quantum} devices are the only quantum devices that can be used in the near-term, where only a limited number of qubits without error-correcting
can be used.

Many quantum machine learning algorithms, such as quantum support vector machine, quantum principal component analysis, and quantum Boltzmann machine, have been developed \cite{wiebe2012quantum, schuld2015introduction, biamonte2017quantum, rebentrost2014quantum, lloyd2014quantum, amin2018quantum, gao2018quantum}, and these algorithms were shown to be more efficient than their classical versions.
Recently, several NISQ quantum machine learning algorithms, such as quantum generative adversarial networks, quantum circuit Born machine, and quantum kernel methods, have been proposed \cite{lloyd2018quantum, dallaire2018quantum, liu2018differentiable, schuld2019quantum, havlivcek2019supervised, benedetti2019parameterized}.
However, these algorithms did not aim to build quantum deep neural networks.

In recent years, deep neural networks \cite{lecun2015deep} became the most important and powerful method in machine learning, which was widely applied in computer vision \cite{voulodimos2018deep}, natural language processing \cite{socher2012deep}, and many other fields. The basic unit of DNN is the perception, which is an affine transformation together with an activation function.
%
%
The non-linearity of the activation function and the depth give the DNN much representation power \cite{HORNIK1991251, LESHNO1993861}.
Approaches have been proposed to build classical DNNs on quantum computers \cite{killoran2018continuous, zhao2019building, kerenidis2020quantum}. They achieved quantum speed-up under certain assumptions.
But the structure of classical DNNs is still used, and only some local operations are speeded up by quantum algorithms, for instance, the inner product was speedup using the swap test~\cite{zhao2019building}. Deep neural networks are also regarded as non-linear functions. In \cite{PhysRevA.98.032309}, they considered the ability of one parameterized quantum circuit (PQC) as a non-linear function.

In this paper, we introduce the quantum neural network layer to classical DNN, the QDNN, which consists of quantum structured layers. It is proved
that the QDNN can uniformly approximate any continuous function
and has more representation power than the classical DNN. Furthermore, it still keeps the advantages of the classical DNN such as the non-linear activation, the multi-layer structure, and the efficient backpropagation training algorithm.
Moreover, the QDNN can be used on near-term noisy intermediate-scale quantum processors.
Therefore, QDNN provides new neural network which can be used in near-term quantum computers
and is more powerful than traditional DNNs.

\section{Methods} 
\label{sec:methods}

The main contribution of this paper is to introduce the concept of the quantum neural network layer (QNNL) as a quantum analog to the classic neural network layer in DNN.
Since all quantum gates are unitary and hence linear, the main difficulty of building a QNNL is introducing non-linearity. We solve this problem by encoding the input vector to a quantum state non-linearly with a
parameterized quantum circuit (PQC), which is similar to \cite{PhysRevA.98.032309}.
Since all quantum gates are unitary and hence linear, the main difficulty of building a QNNL is introducing non-linearity. We solve this problem by encoding the input vector to a quantum state non-linearly with a
parameterized quantum circuit (PQC) \cite{benedetti2019parameterized}.
A QNNL is a quantum circuit that is different from the classical neural network layer. Different from \cite{PhysRevA.98.032309}, we consider a multi-layer structure that may contain multiple PQCs with a deep neural network. A quantum DNN (QDNN) can be easily built with QNNLs since the input and output of a QNNL are classical values.

The advantage of introducing QNNLs is that we can access vectors of exponential dimensional Hilbert spaces with only polynomial resources on a quantum computer.
We prove that this model can not be classically simulated efficiently unless universal quantum computing can be classically simulated efficiently. So QDNNs have more representation power than classical DNNs. We also give training algorithms of QDNNs which are similar to the backpropagation (BP) algorithm. Moreover, QNNLs use the hybrid quantum-classical scheme. Hence, a QDNN with a reasonable size can be trained efficiently on NISQ processors.
Finally, a numerical experiment for image recognition is given using QDNNs, where a high accuracy rate is achieved and only 8 qubits are used.

We finally remark that all tasks using DNN can be turned into
quantum algorithms with more representation powers by replacing the DNN by QDNN.


\subsection{Quantum neural network layers} 
\label{sub:quantum_neural_network_layers}
A deep neural network consists of a large number of {\em neural network layers}, and each neural network layer is a non-linear function
$f_{\overrightarrow{W},\vec{b}}(\vec{x}): \mathbb{R}^n\rightarrow \mathbb{R}^m$
with parameters $\overrightarrow{W},\vec{b}$.
In the classical DNN,  $f_{\overrightarrow{W},\vec{b}}$ takes the form of $\sigma\circ L_{\overrightarrow{W},\vec{b}}$, where  $L_{\overrightarrow{W},\vec{b}}$ is an affine transformation and $\sigma$ is a non-linear activation function. The power of DNNs comes from the non-linearity of the activation function. Without activation functions, DNNs will be nothing more than affine transformations.

However, all quantum gates are unitary matrices and hence linear. So the key point of developing QNNLs is introducing non-linearity.

We build QNNLs using the hybrid quantum-classical algorithm scheme \cite{mcclean2016theory},
which is widely used in many NISQ quantum algorithms \cite{liu2018differentiable,liu2019variational}.
As shown in figure \ref{fig-1}, a hybrid quantum-classical algorithm scheme  consists of a quantum part and a classical part. In the quantum part,
parameterized quantum circuits (PQCs) are used to prepare quantum states with quantum processors.
In the classical part, parameters of the PQCs are optimized using classical computers.
\begin{figure}
    \centering
    \includegraphics{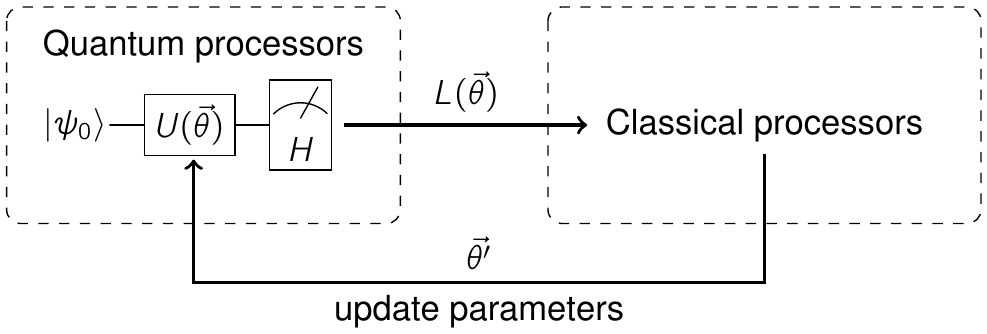}
    \caption{Hybrid quantum-classical scheme. }
\label{fig-1}
\end{figure}

A PQC  is a quantum circuit with parametric gates, which is of the form \[
    U(\vec{\theta}) = \prod_{j=1}^l U_j(\theta_j)
\]
where $\vec{\theta}=(\theta_1,\dots,\theta_l)$ are the parameters, each $U_j(\theta_j)$ is a rotation gate $U_j(\theta_j) = \exp(-i \frac{\theta_j}{2}H_j)$, and  $H_j$ is a 1-qubit or a 2-qubits gate such that $H_j^2 = I$. For example, in this paper we will use the Pauli gates $X,Y,Z$, and the $\mathrm{CNOT}$ gate.

As shown in figure \ref{fig-1}, once fixed an ansatz circuit $U(\vec{\theta})$ and a Hamiltonian $H$, we can define the loss function of the form $L = \bra{0}U^\dagger(\vec{\theta}) H U(\vec{\theta})\ket{0}$. Then we can optimize $L$ by updating parameters $\vec\theta$ using optimizating algorithms \cite{li2017hybrid,nakanishi2019sequential}.
With gradient-based algorithms \cite{li2017hybrid}, one can efficiently get gradient imformation $\frac{\partial L}{\partial \vec{\theta}}$ which is essentially important in our model. Hence, we will focus on gradient-based algorithms in this paper.

Now we are going to define a QNNL, which consists of 3 parts: the encoder, the transformation,
and the output, as shown in figure \ref{fig-2}.

For a classical input data $\vec{x}\in \mathbb{R}^n$, we introduce non-linearity to our QNNL by encoding the input $\vec{x}$ to a quantum state $\ket{\psi(\vec{x})}$ non-linearly. Precisely, we choose a PQC $U(\vec{x})$ with at most $O(n)$ qubits and apply it to an initial state $\ket{\psi_0}$ to obtain a quantum state
\begin{equation}\label{eq-ql1}
\ket{\psi(\vec{x})}=U(\vec{x})\ket{\varphi_0}
\end{equation}
encoded from $\vec{x}$. The PQC is naturally non-linear in the parameters. For example, the encoding process \begin{equation*}
    \ket{\psi(x)}=\exp(-i \frac{x}{2}X)\ket{0}
\end{equation*}
from $x$ to $\ket{\psi(x)}$ is non-linear. Moreover, we can compute the gradient of each component of $\vec{x}$ efficiently.  The gradient of the input in each layer is necessary when training the QDNN. The encoding step is the analog to the classical activation step.

After encoding the input data, we apply a linear transformation as the analog of the linear transformation in the classical DNNs. This part is natural on quantum computers because all quantum gates are linear. We use another PQC $V(\overrightarrow{W})$ with parameters $\overrightarrow{W}$ for this purpose. We assume that the number of parameters in $V(\overrightarrow{W})$ is $O(\mathrm{poly}(n))$.

Finally, the output of a QNNL will be computed as follow. We choose $m$ fixed  Hamiltonians $H_j$, $j=1,\dots,m$,
and output
\begin{equation}\label{eq-ql2}
    \vec{y} = \begin{pmatrix} y_1 + b_1 \\ \vdots \\ y_m + b_m \end{pmatrix},\quad y_j = \bra{\psi(\vec{x})}V^{\dagger}(\overrightarrow{W}) H_j V(\overrightarrow{W})\ket{\psi(\vec{x})}, b_j\in \mathbb{R}.
\end{equation}
Here, the bias term $\vec{b}=(b_1,\dots,b_m)$ is an analog of bias in classical DNNs.
Also, each $y_j$ is a  hybrid quantum-classical scheme with PQC $U$
and Hamiltonian $V^{\dagger}(\overrightarrow{W}) H_j V(\overrightarrow{W})$.

\begin{figure}
\centering
\includegraphics[width=1\textwidth]{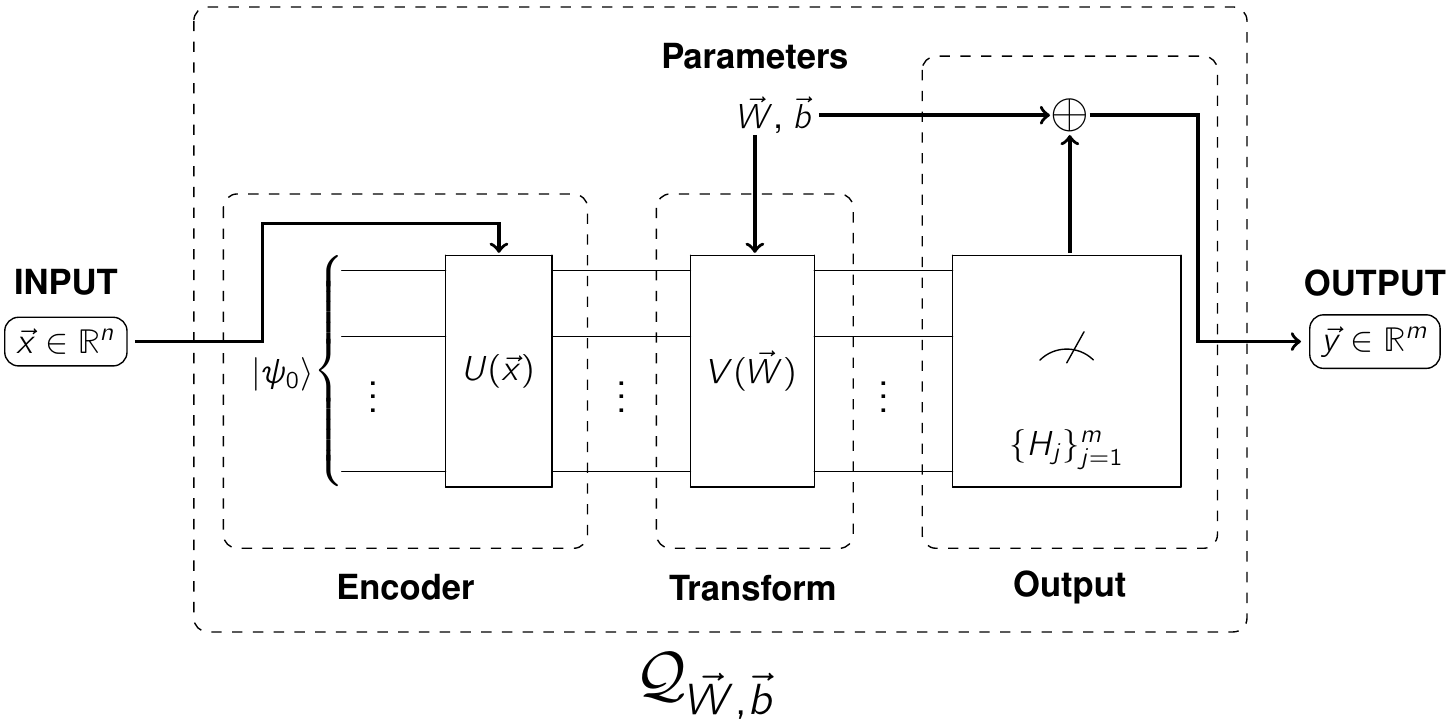}
\caption{The structure of a QNNL $\mathcal{Q}_{\vec{W},\vec{b}}$}
\label{fig-2}
\end{figure}

To compute the output efficiently, we assume that the expectation value of each of these Hamiltonians can be computed in $O(\mathrm{poly}(n,\frac{1}{\varepsilon}))$, where $\varepsilon$ is the precision.
It is easy to show that all Hamiltonians of the following form satisfy this assumption
 \[
    H = \sum_{i=1}^{O(\mathrm{poly}(n))} H_i,
\]
where $H_i$ are tensor products of Pauli matrices or $k$-local Hamiltonians.

In summary, a {\em QNNL} is a function $$\mathcal{Q}_{\overrightarrow{W},\vec{b}}(\vec{x}): \mathbb{R}^n\rightarrow \mathbb{R}^m$$
defined by (\ref{eq-ql1}) and (\ref{eq-ql2}), and shown in figure \ref{fig-2}.
Note that a QNNL is a function with classic input and output,
and can be determined by a tuple \[
    \mathcal{Q} = (U,V,[H_j]_{j=1,\dots,m})
\]
with parameters $(\overrightarrow{W}, \vec{b})$.
Notice that the QNNLs activate before affine transformations while classical DNNLs activate after affine transformations. But this difference can be ignored when considering multi-layers.

\subsection{QDNN and its training algorithms} 
Since the input and output of QNNLs are classical values, the QNNLs can be naturally embedded in classical DNNs. A neural network consists of the composition of multiple compatible QNNLs and classical DNN layers is called {\em quantum DNN (QDNN)}:
$$QDNN =
\mathcal{L}_{1,\overrightarrow{W}_1,\vec{b_1}}\circ \cdots \circ \mathcal{L}_{l,\overrightarrow{W}_l,\vec{b_l}}$$
where each $\mathcal{L}_{i,\overrightarrow{W_i},\vec{b_i}}$ is a classical or a quantum layer
from $\R^{n_{i-1}}$ to $\R^{n_{i}}$ for $i=1,\dots,l$
and $\{\overrightarrow{W_i},\vec{b_i},i=1,\dots,l\}$ are the parameters of the QDNN.

We will use gradient descent to update the parameters. In classical DNNs, the gradient of parameters in each layer is computed by the backpropagation algorithm (BP). Suppose that we have a QDNN. Consider a QNNL $\mathcal{Q}$ with parameters $\vec{W}, \vec{b}$, whose input is $\vec{x}$ and output is $\vec{y}$. Refer to (\ref{eq-ql1}) and (\ref{eq-ql2}) for details.

To use the BP algorithm, we need to compute $\frac{\partial \vec{y}}{\partial \overrightarrow{W}}, \frac{\partial \vec{y}}{\partial \vec{b}}$ and $\frac{\partial \vec{y}}{\partial \vec{x}}$.
Computing $\frac{\partial \vec{y}}{\partial \vec{b}}$ is trivial. Because $U, V$ are PQCs and each component of $\vec{y}$ is an output of a hybrid quantum-classical scheme, both $\frac{\partial \vec{y}}{\partial \overrightarrow{W}}$ and $\frac{\partial \vec{y}}{\partial \vec{x}}$ can be estimated by shifting parameters \cite{li2017hybrid}.

\begin{algorithm}[H]
\label{algo-1}
\SetAlgoLined
\KwInput{A PQC $U(\vec{\theta})$, an initial state $\ket{\varphi_0}$, a Hamiltonian $H$, and an initial value $\vec{x} = (x_1, \dots, x_m)$}
\KwOutput{The gradient $\frac{\partial}{\partial x_j} \bra{\varphi_0}U^\dagger(\vec{x})HU(\vec{x})\ket{\varphi_0}$}

Set $x_j:=x_j+\frac{\pi}{2}$ for each $j$\;
Estimate $\braket{H_{j,+}} = \bra{\varphi_0}U^\dagger(\vec{x})HU(\vec{x})\ket{\varphi_0}$\;
Set $x_j:=x_j - \pi$ for each $j$\;
Estimate $\braket{H_{j,-}} = \bra{\varphi_0}U^\dagger(\vec{x})HU(\vec{x})\ket{\varphi_0}$\;
\Return{$\frac{1}{2}\big[\braket{H_{j,+}} - \braket{H_{j,-}}\big]$
\caption{Gradient estimation for PQCs}}
\end{algorithm}

We can use the above algorithm to estimate the gradient in each quantum layer.

\begin{algorithm}[H]
\SetAlgoLined
\KwInput{A QNNL $\mathcal{Q}(\vec{x},\vec{W},\vec{b})$, the current value of $\vec{x},\vec{W},\vec{b}$ and $\vec{y} = \mathcal{Q}(\vec{x},\vec{W},\vec{b})$, the gradient of the output $\frac{\partial L}{\partial \vec{y}}$}
\KwOutput{The gradient $\frac{\partial L}{\partial \vec{x}}, \frac{\partial L}{\partial \vec{W}}, \frac{\partial L}{\partial \vec{b}}$}

Initialize $\frac{\partial \vec{y}}{\partial \vec{x}}$ and $\frac{\partial \vec{y}}{\partial \vec{W}}$\;
\For{$x_j$ in $\vec{x}$}
{
    Estimate $\frac{\partial \vec{y}}{\partial x_j}$ with Algorithm \ref{algo-1}\;
    Set the $j$-th component of $\frac{\partial \vec{y}}{\partial \vec{x}}$ to $\frac{\partial \vec{y}}{\partial x_j}$\;
}
\For{$w_k$ in $\vec{W}$}
{
    Estimate $\frac{\partial \vec{y}}{\partial w_k}$ with Algorithm \ref{algo-1}\;
    Set the $k$-th component of $\frac{\partial \vec{y}}{\partial \vec{W}}$ to $\frac{\partial \vec{y}}{\partial w_k}$\;
}

Set $\frac{\partial L}{\partial \vec{x}}:=\frac{\partial L}{\partial \vec{y}} \frac{\partial \vec{y}}{\partial \vec{x}}$\;
Set $\frac{\partial L}{\partial \vec{W}}:=\frac{\partial L}{\partial \vec{y}} \frac{\partial \vec{y}}{\partial \vec{W}}$\;
Set $\frac{\partial L}{\partial \vec{b}}:=\frac{\partial L}{\partial \vec{y}}$\;
\Return{$\frac{\partial L}{\partial \vec{x}}, \frac{\partial L}{\partial \vec{W}}, \frac{\partial L}{\partial \vec{b}}$\;
\caption{Gradient estimation for QNNLs}}
\end{algorithm}
Hence, gradients can be back propagated through the quantum layer, and QDNNs can be trained with the BP algorithm.

\subsection{Representation power of QDNNs} 
\label{sub:representation_power_of_qdnns}
In this section, we will consider the representation power of the QDNN. We will show that QDNN can approximate any continuous function similar to the classical DNN. Moreover, if quantum computing can not be classically simulated efficiently, the QDNN has more representation power than the classical DNN with polynomial computation resources.

\subsubsection{Approximate functions with QDNNs}
The universal approximation theorem ensures that DNNs can approximate any continuous function \cite{cybenko1989approximation, HORNIK1991251, LESHNO1993861,pinkus_1999}. Since the class of QDNNs is an extension of the class of classical DNNs, the universal approximation theorem can be applied to the QDNN trivially.
Now, let us consider two subclasses of the QDNN.
\begin{itemize}
\item DNN with only QNNLs.
\item DNN with QNNLs and affine layers.
\end{itemize}

In the first subclass, let us consider a special type of QNNLs which can represent monomials \cite{PhysRevA.98.032309}. Consider the circuit \begin{equation}
    U(x) = R_y(2\arccos(\sqrt{x})) = \begin{pmatrix}
        x & -\sqrt{1-x^2} \\
        \sqrt{1-x^2} & x
    \end{pmatrix}
\end{equation} and the Hamiltonian $H_0 = \ket{0}\bra{0}$. The expectation value $\bra{0}U^\dagger(x)H_0U(x)\ket{0}$ is the monomial $x$ for $x\in[0,1]$. For multivariable monomial $\mathbf{x} = x_1^{m_1}\cdots x_k^{m_k}$, we use the circuit \begin{equation}
    U(\mathbf{x}) = \big[\otimes_{j_1 = 1}^{m_1}R_y(2\arccos(\sqrt{x_1}))\big]\otimes\dots\otimes \big[\otimes_{j_k = 1}^{m_k}R_y(2\arccos(\sqrt{x_k}))\big]
\end{equation} and the Hamiltonian $H_\mathbf{0} = \ket{0\dots 0}\bra{0\dots 0}$. Similarly, the expectation value of
\[
    \bra{0\dots 0}U^\dagger(\mathbf{x})H_{\mathbf{0}}U(\mathbf{x})\ket{0\dots 0} = x_1^{m_1}\cdots x_k^{m_k}
\] for $x_1,\dots,x_k\in[0,1]$.

With the above results and Stone-Weierstrass theorem \cite{10.2307/3029750}, we can deduce the following theorem.
\begin{theorem}
\label{approx-th1}
The QDNN with only QNNLs can uniformly approximate any continuous function \[
    f:[0,1]^k \rightarrow \mathbb{R}^l.
\]
\end{theorem}

Now let us consider the second subclass. As the affine transformation can map the hypercube $[0,1]^k$ to $[a_1,b_1]\times\dots\times [a_k, b_k]$ for any $a_k < b_k$. Hence we have the following result.
\begin{corollary}
\label{approx-co1}
The QDNN with QNNLs and affine layers can uniformly approximate any continuous function \[
    f:D \rightarrow \mathbb{R}^l,
\] where $D$ is a compact set in $\mathbb{R}^k$.
\end{corollary}

Also, the QNNL can   perform as a non-linear activation function. For example, we consider a QNNL $\mathcal{Q}_{\mathrm{ac}}$ with the input circuit \[
    \otimes_{j=1}^m R_y(x_j)
\] and the Hamiltonian
\[
    H_j = I\otimes\dots\otimes I\otimes\ket{0}\bra{0}\otimes I\otimes\dots\otimes I,
\]
where the projection is on the $j$-th qubit for $j = 1,\dots,m$. By simple computation, we have \begin{equation}
    \mathcal{Q}_{\mathrm{ac}}(\begin{pmatrix}
    x_1\\ \vdots\\ x_m
    \end{pmatrix}) = \begin{pmatrix}
    \cos(x_1)\\ \vdots\\ \cos(x_m)
    \end{pmatrix}.
\end{equation}
By the  universal approximation property \cite{LESHNO1993861,kratsios2019universal}, neural networks with non-polynomial activation functions can approximate any continuous function $f:\mathbb{R}^k \rightarrow \mathbb{R}^l$. Thus, the QDNN with QNNLs and affine layers can approximate any continuous function.

\subsubsection{Quantum advantages}

According to the definition of QNNLs in (\ref{eq-ql2}), each element of the outputs in a QNNL is of the form
\begin{equation}
    y_j = b_j + \bra{\psi_0} U^\dagger(\vec{x}) V^{\dagger}(\overrightarrow{W}) H_j V(\overrightarrow{W}) U(\vec{x}) \ket{\psi_0}.
    \label{eq:expection}
\end{equation}
In general, estimation of $\bra{\psi_0} U^\dagger(\vec{x}) V^{\dagger}(\overrightarrow{W}) H_j V(\overrightarrow{W}) U(\vec{x}) \ket{\psi_0}$  on a classical computer will be difficult by the following theorem.
\begin{theorem}
\label{th-1}
Estimation (\ref{eq:expection}) with precision $c<\frac{1}{3}$ is $\BQP$-hard,
where $\BQP$ is the bounded-error quantum polynomial time complexity class.
\end{theorem}
\begin{proof}
Consider any language $L\in \BQP$. There exists a polynomial-time Turing machine which takes $x\in\{0,1\}^n$ as input and outputs a polynomial-sized quantum circuit $C(x)$. Moreover, $x\in L$ if and only if the measurement result of $C(x)\ket{0}$ of the first qubit has the probability $\geq \frac{2}{3}$ to be $\ket{1}$.

Because $\{R_x, R_y, R_z, \mathrm{CNOT}\}$ are universal quantum gates, $C(x)$ can be expressed as a polynomial-sized PQC: $U_x(\vec{\theta}) = C(x)$ with proper parameters. Consider $H = Z\otimes I\otimes \dots\otimes I$, then \begin{equation}
    \bra{0} U_x(\vec{\theta}) H U_x(\vec{\theta}) \ket{0} \leq -\frac{1}{3}
\end{equation} if and only if $x\in L$, and \begin{equation}
    \bra{0} U_x(\vec{\theta}) H U_x(\vec{\theta}) \ket{0} \geq \frac{1}{3}
\end{equation} if and only if $x\notin L$.
\end{proof}

Given inputs, computing the outputs of classical DNNs is polynomial time. Hence, functions represented by classical DNNs are characterized by the complexity class $\Ppoly$. On the other hand, computing the outputs of QDNNs is $\BQP$-hard in general according to Theorem \ref{th-1}. The functions represented by QDNNs are characterized by a complexity class that has a lower bound $\ComplexityFont{BQP/poly}$. Here, $\ComplexityFont{BQP/poly}$ is the problems which can be solved by polynomial sized quantum circuits with bounded error probability \cite{aaronson2005complexity}. Under the hypothesis that quantum computers cannot be simulated efficiently by classical computers, which is generally believed, there exists a function represented by a QDNN which cannot be computed by classical circuits of polynomial size.
Hence, QDNNs have more representation power than DNNs.

\section{Results} 
\label{sec:results}
We used QDNN to conduct a numerical experiment for an image classification task. The data comes from the MNIST data set. We built a QDNN with 3 QNNLs. The goal of this QDNN is to recognize the digit in the image is either 0 or 1 as a classifier.

\subsection{Experiment details} 
The data in the MNIST is $28\times 28 = 784 $ dimensional images. This dimension is too large for the current quantum simulator. Hence, we first resize the image to $8\times 8$ pixels. We use three QNNLs in our QDNN, which will be called the input layer, the hidden layer, and the output layer, respectively.

In each layer, we set the ansatz circuit to be the one in figure \ref{fig:ansatz}, which is similar to the hardware efficient ansatz \cite{kandala2017hardware}.

\begin{figure*}
\begin{subfigure}{.4\textwidth}
  \centering
  \includegraphics[width=\linewidth]{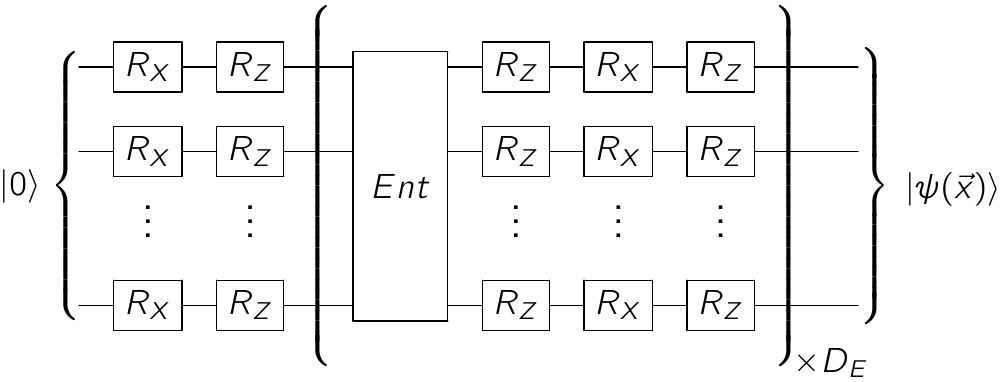}
  \caption{}
  \label{fig:encoder}
\end{subfigure}
\begin{subfigure}{.4\textwidth}
  \centering
  \includegraphics[width=\linewidth]{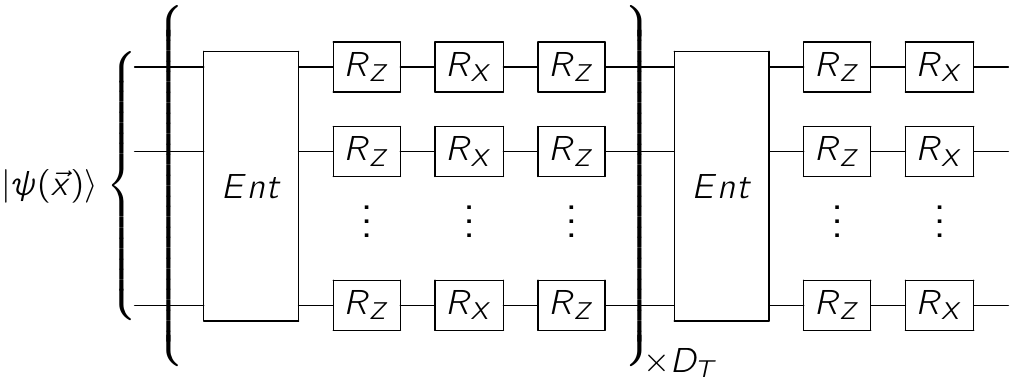}
  \caption{}
  \label{fig:transformation}
\end{subfigure}
\begin{subfigure}{.18\textwidth}
  \centering
  \includegraphics[width=.8\linewidth]{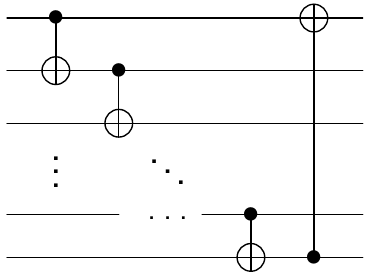}
  \caption{}
  \label{fig:ent}
\end{subfigure}
\caption{Ansatz circuits in each QNNL. a. The ansatz circuit of the encoder. b. The ansatz circuit of the transformation. c. The structure of gate $Ent$.}
\label{fig:ansatz}
\end{figure*}

\subsubsection{Input layer}
The input layer uses an 8-qubit circuit which accept an input vector $x\in\mathbb{R}^{64}$ and output a vector $\vec{h_1}\in\mathbb{R}^{24}$. The structure in FIG. \ref{fig:ansatz} is used, where $D_E=2, D_T=5$. We denote $V_{\mathrm{in}}(\overrightarrow{W}_{\mathrm{in}})$ for $\overrightarrow{W}_{\mathrm{in}}\in \mathbb{R}^{160}$ to be the transformation circuit in this layer.

The output of the input layer is of the form
\begin{equation} \vec{h}_1 =
    \begin{pmatrix}
        \bra{\psi(\vec{x},\overrightarrow{W}_{\mathrm{in}})}H_{1,X}\ket{\psi(\vec{x},\overrightarrow{W}_{\mathrm{in}})} \\
        \vdots \\
        \bra{\psi(\vec{x},\overrightarrow{W}_{\mathrm{in}})}H_{8,X}\ket{\psi(\vec{x},\overrightarrow{W}_{\mathrm{in}})} \\
        \bra{\psi(\vec{x},\overrightarrow{W}_{\mathrm{in}})}H_{1,Y}\ket{\psi(\vec{x},\overrightarrow{W}_{\mathrm{in}})} \\
        \vdots \\
        \bra{\psi(\vec{x},\overrightarrow{W}_{\mathrm{in}})}H_{8,Y}\ket{\psi(\vec{x},\overrightarrow{W}_{\mathrm{in}})} \\
        \bra{\psi(\vec{x},\overrightarrow{W}_{\mathrm{in}})}H_{1,Z}\ket{\psi(\vec{x},\overrightarrow{W}_{\mathrm{in}})} \\
        \vdots \\
        \bra{\psi(\vec{x},\overrightarrow{W}_{\mathrm{in}})}H_{8,Z}\ket{\psi(\vec{x},\overrightarrow{W}_{\mathrm{in}})}
    \end{pmatrix} + \vec{b}_{\mathrm{in}} \in \mathbb{R}^{24},
\end{equation}
where $\ket{\psi(\vec{x},\overrightarrow{W}_{\mathrm{in}})} = V_{\mathrm{in}}(\overrightarrow{W}_{\mathrm{in}})\ket{\psi(\vec{x})}$ and $H_{j,M}$ denotes the result obtained by applying  the operator $M$ on the $j$-th qubit for $M\in\{X,Y,Z\}$.

\subsubsection{Hidden layer}
The hidden layer uses 6 qubits. It accepts an vector $\vec{h_1}\in \mathbb{R}^{24}$ and outputs a vector $\vec{h_2}\in \mathbb{R}^{12}$. The structure shown in figure 4 is used, with $D_E=1, D_T=4$. Because there are $30$ parameters in the encoder, we set the last column of $R_Z$ gates to be $R_Z(0)$.
Similar to the input layer, the output of the hidden layer is \begin{equation} \vec{h_2} =
    \begin{pmatrix}
        \bra{\psi(\vec{h}_1,\overrightarrow{W_h})}H_{1,Y}\ket{\psi(\vec{h}_1,\overrightarrow{W_h})} \\
        \vdots \\
        \bra{\psi(\vec{h}_1,\overrightarrow{W_h})}H_{6,Y}\ket{\psi(\vec{h}_1,\overrightarrow{W_h})} \\
        \bra{\psi(\vec{h}_1,\overrightarrow{W_h})}H_{1,Z}\ket{\psi(\vec{h}_1,\overrightarrow{W_h})} \\
        \vdots \\
        \bra{\psi(\vec{h}_1,\overrightarrow{W_h})}H_{6,Z}\ket{\psi(\vec{h}_1,\overrightarrow{W_h})}
    \end{pmatrix} + \vec{b_h} \in \mathbb{R}^{12}.
\end{equation}

\subsubsection{Output layer}
The output layer uses 4 qubits. We also use the structure in FIG. \ref{fig:ansatz} with $D_E=1, D_T=2$. Because there are $20$ parameters in the encoder, we set the last column of $R_Z$ and $R_X$ gates to be $R_Z(0)$ and $R_X(0)$.
The output of the output layer is \begin{equation}
    \vec{y} =
    \begin{pmatrix}
        \bra{\psi(\vec{h_2},\overrightarrow{W}_\mathrm{out})}(\ket{0}\bra{0}\otimes I\otimes I\otimes I) \ket{\psi(\vec{h_2},\overrightarrow{W}_\mathrm{out})} \\
        \bra{\psi(\vec{h_2},\overrightarrow{W}_\mathrm{out})}(\ket{1}\bra{1}\otimes I\otimes I\otimes I) \ket{\psi(\vec{h_2},\overrightarrow{W}_\mathrm{out})}.
    \end{pmatrix}
\end{equation}
Notice that we do not add bias term here, and it will output a vector in $\mathbb{R}^2$. Moreover, after training, we hope to see if the input $\vec{x}$ is from an image of digit $0$, the output $\vec{y}$ should be close to $\ket{0}$, otherwise it should be close to $\ket{1}$.

In conclusion, the settings of these three layers are shown in table \ref{tab:table1}.

Finally, the loss function is defined as \begin{equation}
    L = \frac{1}{|\mathcal{D}|}\sum_{(\vec{x},y)\in \mathcal{D}}\big|\mathrm{DNN}(\vec{x}) - \ket{y}\big|^2,
\end{equation}
where $\mathcal{D}$ is the training set.

\newcommand{\tabincell}[2]{\begin{tabular}{@{}#1@{}}#2\end{tabular}}
\begin{table}
    \centering
    \begin{tabular}{l|p{1.2cm}<{\centering}|p{1.8cm}<{\centering}|p{1.8cm}<{\centering}|p{3cm}<{\centering}}
         & \# of qubits & Input dimension & Output dimension & \# of parameters (transformation + bias) \\
        \hline
        Input layer & 8 & 64 & 24 & 136 + 24\\
        Hidden layer & 6 & 24 & 12 & 84 + 12\\
        Output layer & 4 & 12 & 2 & 32 + 0
    \end{tabular}
    \caption{Settings of three layers}
    \label{tab:table1}
\end{table}

\subsection{Experiments result}
We used the Julia package \texttt{Yao.jl} \cite{YaoFramework2019} as a quantum simulator in our experiments. All data were collected on a desktop PC with Intel CPU i7-4790 and 4GB RAM.

All parameters were initialized randomly in $(-\pi,\pi)$. We use Adam optimizer \cite{kingma2014adam} to update parameters. We train this QDNN for 400 iterations with batch size of 240. In the first 200 of iterations, the hyper parameters of Adam is set to be $\eta=0.01, \beta_1=0.9, \beta_2=0.999$.
In the later 200 of iterations, we change $\eta$ to $0.001$.

The values of the loss function  on the training set and test set during training is shown in figure \ref{fig:loss}. The accurate rate of this QDNN on the test set rise to $99.29\%$ after training.
\begin{figure}
    \centering
    \includegraphics[width=0.7\textwidth]{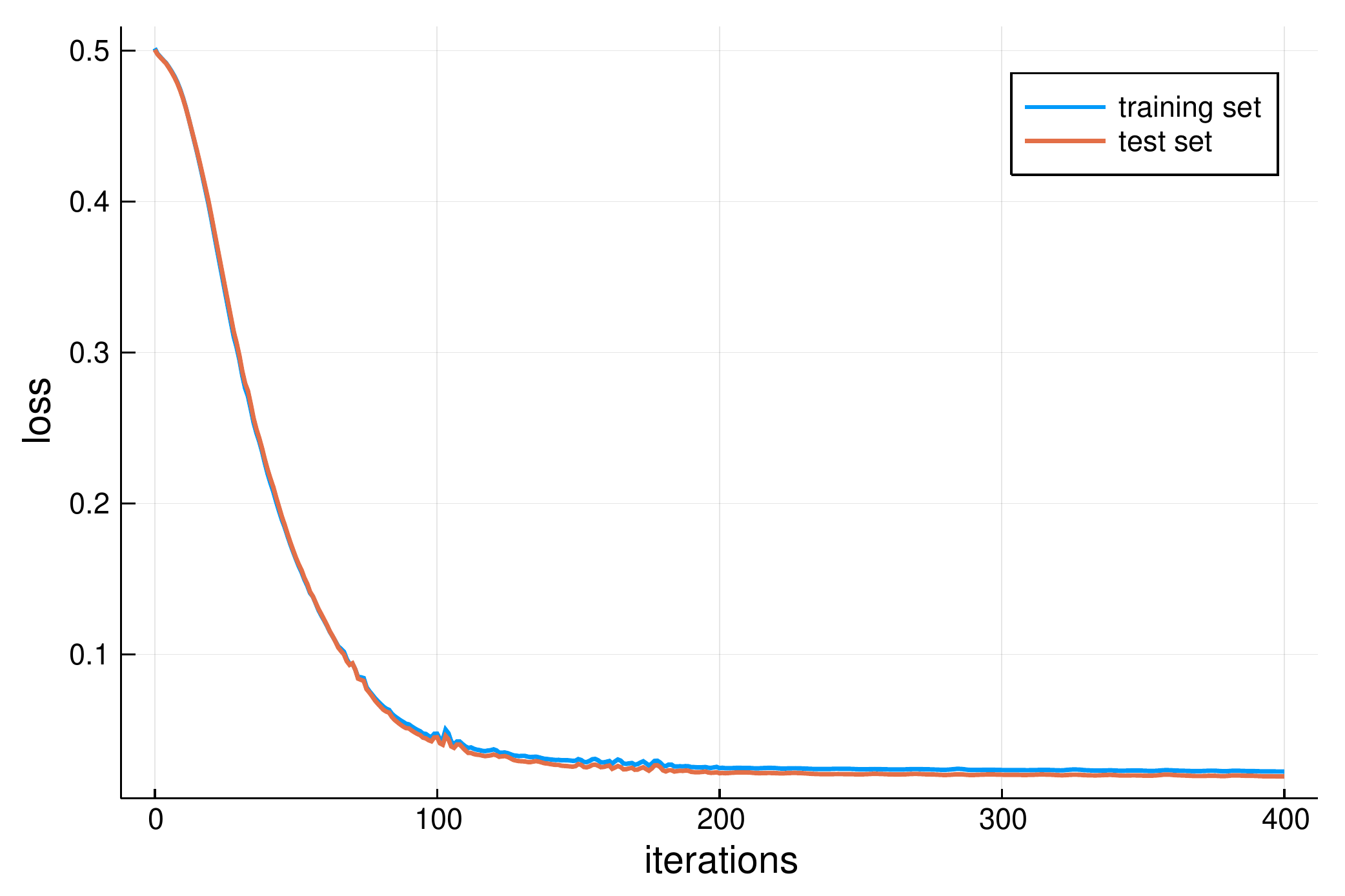}
    \caption{Loss function}
    \label{fig:loss}
\end{figure}

\section{Discussion} 
\label{sec:discussion}
We introduced the model of QNNL and built QDNN with QNNLs. We proved that QDNNs have more representation power than classical DNNs. We presented a practical gradient-based training algorithm as the analog of BP algorithms. Because the model is based on the hybrid quantum-classical scheme, it can be realized on NISQ processors.
As a result, the QDNN has more representation powers than classical DNNs and still
keeps most of the advantages of the classical DNNs.

Since we use a classical simulator on a desktop PC for quantum computation, only QDNNs with a small number of qubits can be used and only simple examples can be demonstrated.
Quantum hardware is developing fast. Google achieved quantum supremacy by using a superconducting quantum processor with 53 qubits \cite{arute2019quantum}.
From table \ref{tab:table1}, up to 8 qubits are used in our experiments described in Sections 2.5 and 4.2,
so in principle, our image classification experiment can be implemented in Google's
quantum processor. With quantum computing resources, we can access exponential dimensional feature Hilbert spaces \cite{schuld2019quantum} with QDNNs and only use polynomial-size of parameters. Hence, we believe that QDNNs will help us to extract features more efficiently than DNNs.
This is similar to the ideas of kernel methods \cite{shawe2004kernel, zelenko2003kernel}.

\section*{Acknowledgements} 
\label{sec:acknowledgements}
We thank Xiu-Zhe Luo and Jin-Guo Liu for helping in Julia programming. This work is partially supported by an NSFC grant no. 11688101 and an NKRDP grant no. 2018YFA0306702.

\section*{Data availability}
All codes are available on \texttt{http://github.com/ChenZhao44/QDNN.jl}. And the data that support the findings of this study are available if request.

\bibliographystyle{naturemag}

\bibliography{QDNN}

\end{document}